\documentclass[conference]{IEEEtran}
\usepackage{amsmath,amsfonts,amssymb}
\usepackage{algorithmic}
\usepackage{algorithm}
\usepackage{amsthm}
\usepackage{array}
\usepackage[caption=false]{subfig}
\usepackage{textcomp}
\usepackage{stfloats}
\usepackage{mathtools, cuted}
\usepackage{url}
\usepackage[utf8]{inputenc}
\usepackage[T1]{fontenc}
\usepackage{hyperref}
\usepackage{verbatim}
\usepackage{graphicx}
\usepackage[noadjust]{cite}
\usepackage[thinc]{esdiff}
\usepackage{xcolor}
\DeclareMathOperator{\Exp}{\mathbb{E}}
\newtheorem{theorem}{Theorem}
\newtheorem{lemma}[theorem]{Lemma}
\newtheorem{remark}{Remark}

\ifCLASSINFOpdf
\else
\fi

\begin{document}
%
\title{Age Performance Analysis in Resource-Constrained Adversarial IoT Systems}
%
%
%


\author{
    \IEEEauthorblockN{
	Aresh Dadlani\IEEEauthorrefmark{1}\IEEEauthorrefmark{3}, 
	Muthukrishnan Senthil Kumar\IEEEauthorrefmark{2}, 
	Omid Ardakanian\IEEEauthorrefmark{3}, and
	Ioanis Nikolaidis\IEEEauthorrefmark{3}
    }
    \IEEEauthorblockA{
        \IEEEauthorrefmark{1}School of Computing Sciences and Mathematics, Mount Royal University, Calgary, Canada\\
	    \IEEEauthorrefmark{2}Department of Applied Mathematics and Computational Sciences, PSG College of Technology, Coimbatore, India\\
	    \IEEEauthorrefmark{3}Department of Computing Science, University of Alberta, Edmonton, Canada\\
	Emails: adadlani@mtroyal.ca, msk.amcs@psgtech.ac.in, \{ardakanian,  nikolaidis\}@ualberta.ca
    }
}


%
%

\maketitle

\begin{abstract}
\fontdimen2\font=0.64ex
Timely updates are critical for real-time monitoring and control applications powered by the Internet of Things~(IoT). As these systems scale, they become increasingly vulnerable to adversarial attacks, where malicious agents interfere with~legitimate transmissions to reduce data rates, thereby inflating~the age of information (AoI). Existing adversarial AoI models often assume stationary channels and overlook queuing dynamics arising from compromised sensing sources operating under resource constraints. Inspired by the G-queue framework, this paper investigates a two-source M/G/1/1 system in which one source is adversarial and disrupts the update process by injecting negative arrivals according to a Poisson process and inducing i.i.d. service slowdowns, bounded in attack rate and duration.~Using moment generating functions, we then derive closed-form expressions~for average and peak AoI for an arbitrary number~of sources. Moreover, we introduce a worst-case constrained attack model and employ stochastic dominance arguments to establish analytical AoI bounds. Numerical results validate the analysis~and~show~the impact of resource-limited adversarial interference across different service-time distributions.
\end{abstract}
\vspace{0.4em}


%
\IEEEpeerreviewmaketitle

\begin{figure*}[t]
    \centering
  \subfloat[\label{fig1a}]{%
       \includegraphics[width=0.97\columnwidth]{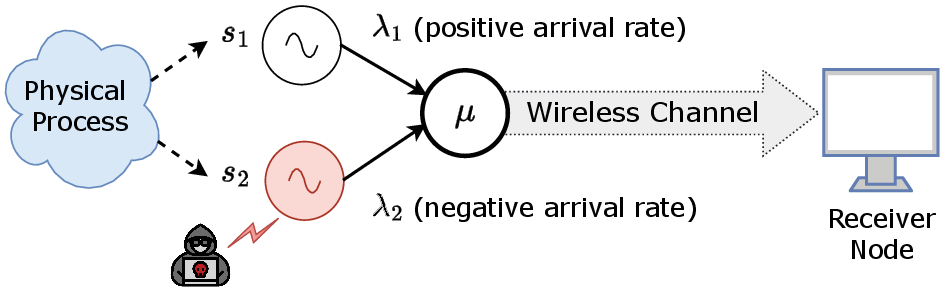}}
    ~\qquad~
  \subfloat[\label{fig1b}]{%
        \includegraphics[width=0.85\columnwidth]{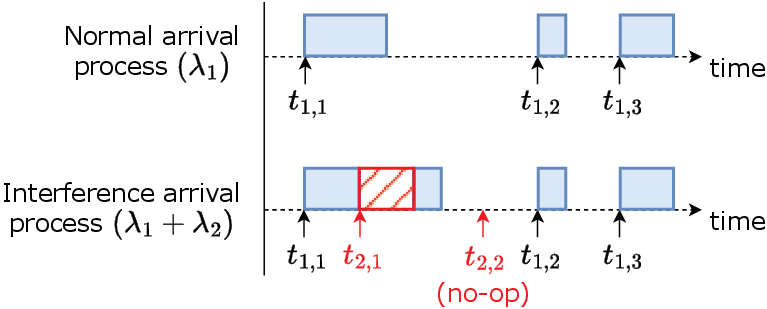}} 
  \caption{Schematic of a real-time monitoring system under attack: (a) Two IoT sensing devices with one compromised~$(s_2)$. (b) Timeline showing the impact of negative arrivals on the service times of status updates. Idle injections ($t_{2,2}$) have no effect (no-op), while collisions ($t_{2,1}$) preempt updates, causing delayed service and increased AoI.}
  \label{fig_1} 
\end{figure*}
\section{Introduction}
\label{sec_1}
\fontdimen2\font=0.7ex
The convergence of Internet of Things (IoT) platforms with low-latency 5G connectivity has enabled a new class of real-time applications such as environmental monitoring, industrial automation, and remote healthcare~\cite{Chettri2020}. These systems rely on distributed sensing devices to generate time-stamped status updates for analysis and control at a receiver. In such time-sensitive settings, the \textit{freshness} of received information is critical for maintaining system responsiveness, control stability, and quality of experience. Traditional performance metrics, such as throughput and delay, fail to capture this notion of timeliness from the perspective of the monitor. To address~this requirement, the \textit{age of information}~(AoI) metric has been~introduced to quantify the time elapsed since the most recently received update was generated~\cite{AbdElmagid2019, Yates2021, Kahraman2024}.

Minimizing AoI in IoT systems is an active research area that requires a nuanced understanding of the dynamics that govern status delivery. To model the inherent delays caused by physical separation, channel contention, and access protocols, the communication link is often represented as a queuing~system. Early works studied how contention and scheduling~affect AoI in single-source settings under various service~disciplines and access strategies~\cite{PappasBook2023}, and subsequent studies extended~these models to multi-source~\cite{Moltafet2020, Moradian2020, Moltafet2022, Moradian2024} and multi-hop~\cite{Chiariotti2022, Sinha2024, Kumar2025} networks with heterogeneous update generation, buffer constraints, and packet management mechanisms.

A growing challenge in large-scale IoT deployments is~their vulnerability to malicious attacks that exploit protocol properties and flaws to degrade information freshness. Such disruptions are particularly detrimental in real-time applications, where delayed or outdated updates can compromise safety~and control performance. Much of the existing literature focuses~on optimizing AoI under wireless channel conditions with additive noise, fading, and unintentional interference. In~\cite{Banerjee2022}, a constrained adversary jams a limited fraction of transmission slots, while~\cite{Sinha2022} develops online scheduling policies~for adversarial and stochastic wireless environments. Game-theoretic strategies for adversarially jammed status update~systems are studied in~\cite{Banerjee2024}. While insightful, these efforts~(\cite{Banerjee2022, Sinha2022, Banerjee2024})~do not explicitly examine adversarial interference from~a~queuing-theoretic perspective. In practice, adversaries may be constrained by transmission power or spectrum access. For example, an attacker using directional interference to disrupt~a~communication channel may be unable to consistently track a~moving receiver and therefore may interfere only sporadically.~The impact of such constrained adversarial (negative) arrivals, modeled within the G-queue framework~\cite{Gelenbe1991}, on AoI remains largely unexplored. A related study modeled negative Poisson arrivals in an M/M/1/1 system that remove the packet currently in the system~\cite{Doncel2025}, while subsequent work investigated whether attacks can reduce AoI in related queuing settings~\cite{DoncelAssaad2025PE}. Nonetheless, these analyses do not extend to non-Markovian service time distributions or service slowdowns under physical constraints limiting sustained attacks.

This paper analyzes AoI under adversarial interference, modeled as negative arrivals, in a two-source M/G/1/1 system. We first derive the Laplace-Stieltjes transforms (LSTs)~of stationary AoI and peak AoI, from which the average AoI~(AAoI) and average peak AoI~(PAoI) are obtained. We then extend the model to multi-source M/G/1/1 queues under similar adversarial conditions, and investigate AAoI bounds under~constraints on the attack rate and slowdown factor. We assume the~attacker cannot continuously interfere with the system due~to limited battery, spectrum access restrictions, or the risk of detection. This constraint ensures a physically plausible attack model and allows for meaningful AoI performance guarantees. Finally, simulation results demonstrate the effects of adversarial interference under general service time distributions and independent and identically distributed (i.i.d.) service slowdowns.

\section{System Model and Definitions}
\label{sec_2}
We consider the real-time IoT monitoring system illustrated in \figurename{~\ref{fig1a}}, where a benign sensor~$s_1$ generates time-stamped (positive) status updates and a compromised sensor~$s_2$ injects disruptive (negative) traffic over a wireless channel. The system is modeled as~a bufferless M/G/1/1 queue that accommodates at most one~update at a time. An update enters service~immediately upon generation if the server is idle. Positive and negative arrivals follow independent Poisson processes with~rates $\lambda_1>0$ and $0\leq\lambda_2\leq\lambda_{\max}$, respectively. Negative arrivals, though carrying no useful information, distort the channel by preempting any positive packet in service. Such disruptions~can arise from actions like jamming, reducing the channel bit-rate, or triggering retransmissions, effectively prolonging the update reception time. Negative arrivals when the server is idle have no effect. This mechanism is motivated by the G-queue framework but, unlike the packet-removal model in~\cite{Doncel2025}, does not discard the update in service.

An illustrative example of this process is depicted in~\figurename{~\ref{fig1b}}. Let $t_{i,j}$ denote the generation time, equivalently the arrival time at the server, of the $j$-th update from source $s_i$, and~let $t'_{i,j}$ be its reception time at the receiver. Under normal operation, updates generated by source $s_1$ arrive at the server at time instants $t_{1,j}$, where $j \!\in\! \{1,2,3\}$. For the negative arrivals at $t_{2,j'}$, where $j' \!\in\! \{1,2\}$, the in-service preemptions $(t_{2,1})$ force retransmissions or re-encoding, whereas the idle injections $(t_{2,2})$ have no effect (no-op). Such interruptions elongate effective service times by inducing service slowdowns for subsequent update transmissions, increase inter-delivery intervals, and inflate AoI.

Let $S_{\mathrm{n}}$ be a continuous random variable~(r.v.) denoting the normal service time of an update from source~$s_1$, with probability density function (pdf) $f_{S_{\mathrm{n}}}(t)$, cumulative distribution function (cdf) $F_{S_{\mathrm{n}}}(t)=\Pr(S_{\mathrm{n}}\leq t)$, and moment generating function (MGF) $M_{S_{\mathrm{n}}}(s)=\mathbb{E}[e^{sS_{\mathrm{n}}}]$. Slow service is parameterized by $S_{\mathrm{s}}=\beta S_{\mathrm{n}}$, where $1<\beta\leq\beta_{\max}$. The recovery time from the slow state is assumed exponential with mean~$1/\Lambda$, where $\Lambda\triangleq\lambda_1+\lambda_2$, after which normal service resumes. The slowdown factor, $1 \!<\! \beta \!\leq\! \beta_{\max}$, models~the~operational limits on the adversary and ensures that the server cannot be forced into a persistently degraded state.

AoI is the time elapsed since the generation of the most~recently received update at the destination, making it a fundamental measure of information freshness. At any given~time~$\tau$, the most recently received update is indexed by $N_{i}(\tau) \!\triangleq\! \max \left\{ j \mid t_{i,j}' \le \tau \right\}$~\cite{Moltafet2020}. Typically exhibiting a sawtooth pattern, the instantaneous AoI of $s_i$ at time~$t$ is defined by~the~random process $A_i(t) \!\triangleq\! t - \psi_i(t)$, where $\psi_i(\tau) \!=\! t_{i, N_i(\tau)}$~is the generation time of the latest received update from $s_i$.~Consequently, the AAoI of source $s_i$ over time interval $(0,\mathcal{T})$ is:
\begin{equation}
    \bar{\Delta}_i = \frac{1}{\mathcal{T}} \int_0^{\mathcal{T}} A_i(t) dt\, .
    \label{eq_1}
\end{equation}

A more mathematically tractable metric is the PAoI, which captures the maximum AoI value immediately before the reception of a new update, and is given as $P_{i,j} = t_{i,j}' - t_{i,j-1}$~\cite{PappasBook2023}. PAoI essentially quantifies the maximum staleness experienced between consecutive successfully delivered updates.

\section{AoI and PAoI in Adversarial M/G/1/1 Queues}
\label{sec_3}
\fontdimen2\font=0.66ex
In this section, we analyze the age metrics in an M/G/1/1 queue subject to adversarial arrivals. To facilitate our analysis, we relate AoI with two key random variables associated with the $j$-th successfully delivered update from source $s_i$. The \textit{system time}~$(T_{i,j})$ is the total time from the generation of an update until its reception at the destination and is calculated~as $T_{i,j} = t_{i,j}' \!-\! t_{i,j}$. The \textit{inter-departure time}~$(Y_{i,j})$ is the time between the reception of the $(j-1)$-th and the $j$-th delivered updates, and is given as $Y_{i,j} = t_{i,j}' - t_{i,j-1}'$. Therefore, the~PAoI of the $j$-th update can be expressed as:
\begin{equation}
    P_{i,j} = Y_{i,j} + T_{i,j-1}.
    \label{eq_4}
\end{equation} 

Assuming the system is stationary and the AoI process~is ergodic, we drop the subscript $j$ in derivations that follow. Therefore, for source~$s_i$, we denote the system and inter-departure times by $T_i$ and $Y_i$, respectively, with MGFs $M_{T_i}(s) \triangleq \mathbb{E}[e^{sT_i}]$ and $M_{Y_i}(s) \triangleq \mathbb{E}[e^{sY_i}]$.

\subsection{Age in the Two-Source Adversarial Model}
\label{sec_3_1}
\fontdimen2\font=0.66ex
For the two-source case introduced in Section~\ref{sec_2}, let 
$A_i$ and $P_i$ denote, respectively, the AoI and PAoI of updates from source~$s_i$, where $i \!\in\! \{1,2\}$, with MGFs $M_{A_i}(s) \triangleq \mathbb{E}\!\left[ e^{s A_i} \right]$ and $M_{P_i}(s) \triangleq \mathbb{E}\!\left[ e^{s P_i} \right]$. Following the general M/G/1/1 AoI characterization in~\cite{Moltafet2022}, the MGFs of AoI and PAoI for the benign source~$s_1$ are given as:
\begin{align}
    M_{A_1}(s) &= \frac{M_{T_1}(s)\,\big(M_{Y_1}(s) - 1\big)}{s\, M_{Y_1}'(0)}\, , 
    \label{eq_5} \\
    M_{P_1}(s) &= M_{T_1}(s)\, M_{Y_1}(s)\, ,
    \label{eq_6}
\end{align}
where $M_{Y_1}'(0)$ denotes the first derivative of $M_{Y_1}(s)$ at $s = 0$. To obtain closed-form expressions for~\eqref{eq_5} and \eqref{eq_6}, we require explicit derivations of $M_{T_1}(s)$ and $M_{Y_1}(s)$. 
\vspace{0.2em}

\subsubsection{Derivation of $M_{T_1}(s)$}
\label{sec_3_1_1}
With $T_1$ representing the system time of a \emph{successfully delivered} update from source~$s_1$, let $\mathcal{D}$~be the event that a positive update completes service before any negative arrival from source~$s_2$. Therefore, for an infinitesimal $\epsilon>0$, the distribution of $T_1$ satisfies~\cite{Najm2018}:
\begin{align}
    \Pr\!\big(T_{1\!} \!\in\! [t, t + \epsilon)\big)\! 
    &=\! \Pr\!\big(S_{\mathrm{n}\!} \!\in\! [t, t + \epsilon) \!\mid\! \mathcal{D}\big) \nonumber \\
    &=\! \frac{\Pr\!\big(S_{\mathrm{n}\!} \!\in\! [t, t \!+\! \epsilon)\big) \Pr\!\big(\mathcal{D} \!\mid\! S_{\mathrm{n}\!} \!\in\! [t, t \!+\! \epsilon)\big)}{\Pr(\mathcal{D})}.
    \label{eq_7}
\end{align}
Taking the limit of \eqref{eq_7} yields the pdf of $T_1$ as follows:
\begin{equation}
    f_{T_1}(t) =\lim_{\epsilon \to 0} \frac {\Pr\! \big(T_{1\!} \in [t, t + \epsilon)\big)}{\epsilon} =  \frac{f_{S_{\mathrm{n}}\!}(t)  \Pr(\mathcal{D} \mid S_{\mathrm{n}\!} > t)}{\Pr(\mathcal{D})}.
    \label{eq_8}
\end{equation}
Since the event $\mathcal{D}$ occurs only if no negative arrivals occur~during the service period of length $t$, we have:
\begin{align}
    &\Pr(\mathcal{D} \mid S_{\mathrm{n}\!} > t) = e^{-\lambda_2 t}\, ,
    \label{eq_9} \\
    \text{and}\quad &\Pr(\mathcal{D}) = \int_0^{\infty} e^{-\lambda_2 t} f_{S_{\mathrm{n}\!}}(t)\, dt = M_{S_{\mathrm{n}}}(-\lambda_2) \, .
    \label{eq_10}
\end{align}
Therefore, plugging \eqref{eq_9} and \eqref{eq_10} into \eqref{eq_8} results in $f_{T_1}(t) = f_{S_{\mathrm{n}}}(t) \, e^{-\lambda_2 t}/M_{S_{\mathrm{n}}}(-\lambda_2)$, which directly leads to the following:
\begin{equation}
    M_{T_1}(s) = \frac{M_{S_{\mathrm{n}}}(s - \lambda_2)}{M_{S_{\mathrm{n}}}(-\lambda_2)}\, .
    \label{eq_12}
\end{equation}
\vspace{0.1em}

\subsubsection{Derivation of $M_{Y_1}(s)$}
\label{sec_3_1_2}
\fontdimen2\font=0.7ex
Departing from the preemptive M/G/1/1 analysis in~\cite{Moltafet2022}, our derivation of $M_{Y_1}(s)$ accounts for the adversarial impact of negative arrivals on the service time distribution. We model $Y_1$, the inter-departure time between two consecutive \emph{delivered} updates from source~$s_1$, using the semi-Markov chain shown in Fig.~\ref{fig_2}. Here, $q_0$, $q_1$, and $q_2$ denote the possible states of the system. Specifically, $q_0$ represents the idle state in which the system awaits the arrival of a fresh positive update, $q_1$ corresponds to the state where a positive update is being served under normal conditions, and $q_2$ denotes the state in which, following a preemption event by a negative arrival, the positive update is served at a reduced (slow) service rate. Upon successful delivery of an $s_1$ update, the system regenerates in state~$q_0$. These state transitions are formalized as follows:
\begin{itemize}
    \item \textbf{$q_0 \!\rightarrow\! q_1$:} In state~$q_0$, the system remains idle and awaits the arrival of a positive update. Upon such an arrival, it deterministically transitions to $q_1$ with probability~$1$, where the update immediately begins normal service. The sojourn time in $q_0$, denoted by $W_1$, follows an exponential distribution with rate $\lambda_1$, i.e., $W_1 \sim \mathrm{Exp}(\lambda_1)$.
    
    \item \textbf{$q_1 \!\rightarrow\! q_0$:} In state $q_1$, the system serves the positive update under normal conditions. If the service completes without preemption, the state transitions back to $q_0$ with probability $p_{\mathcal{D}} \triangleq \Pr(\mathcal{D}) = M_{S_{\mathrm{n}}}(-\lambda_2)$. The sojourn~time in $q_1$ prior to this transition, denoted $W_2$, has the same distribution as the system time $T_1$ conditioned on successful delivery, i.e., $\Pr(W_2 > t)= \Pr(S_{\mathrm{n}\!} > t \mid \mathcal{D})$.
    
    \item \textbf{$q_1 \!\rightarrow\! q_2$:} If there is a negative arrival while a positive update is in normal service, the service is preempted,~and the system transitions from $q_1$ to $q_2$ with probability $p_{\overline{\mathcal{D}}} \triangleq 1 - \Pr(\mathcal{D}) = 1 \!-\! M_{S_{\mathrm{n}}}(-\lambda_2)$. The sojourn time of the system in $q_1$ preceding this transition, denoted by $W_3$,~is distributed as $T_1$ conditioned on negative preemption, i.e., $\Pr(W_3 > t) = \Pr(S_{\mathrm{n}} > t \mid \overline{\mathcal{D}})$.

    \item \textbf{$q_2 \!\rightarrow\! q_1$:} If the slowdown period ends before the update in $q_2$ completes service, the system switches back to $q_1$, resuming normal service. Let $\mathcal{F}$ denote the event that the service completes before a negative arrival from source $s_2$, and $\overline{\mathcal{F}}$ its complement. The probability of completion is given as $p_{\mathcal{F}} \!\triangleq\! \Pr(\mathcal{F}) = M_{S_{\mathrm{s}}}(-\Lambda)$. The sojourn~time in $q_2$ before this transition, denoted $W_4$, is given as $\Pr(W_4 > t) = \Pr(S_{\mathrm{s}\!} > t \mid \mathcal{F})$.

    \item \textbf{$q_2 \!\rightarrow\! q_2$:} When in state $q_2$, the system serves a positive update at the slow rate. This transition from state $q_2$ to itself corresponds to event $\overline{\mathcal{F}}$, i.e., a new negative arrival before service completion, which causes the slow service to be preempted and restarted with probability $p_{\overline{\mathcal{F}}} \!\triangleq\! 1 - \Pr(\mathcal{F}) \!=\! 1 \!-\! M_{S_{\mathrm{s}}}(-\Lambda)$. The sojourn time in $q_2$ prior to this transition, denoted as $W_5$, is determined similar to $W_3$ (i.e., the $q_1 \rightarrow q_2$ transition), and is detailed in Lemma~\ref{lemma_1}.     
\end{itemize}
\begin{figure}[!t]
    \centering
    \includegraphics[width=0.65\columnwidth]{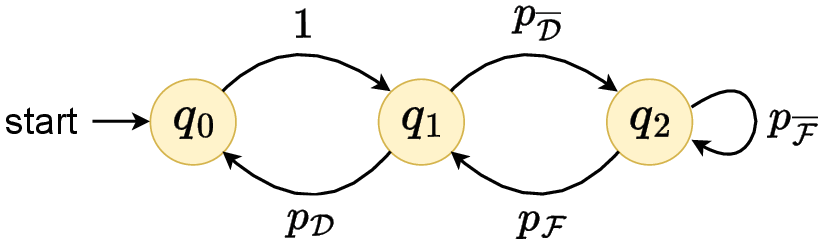}
    \vspace{-0.5em}
    \caption{Semi-Markov chain for the M/G/1/1 inter-departure time $Y_1$.}
    \label{fig_2}
\end{figure}

\begin{lemma}
\label{lemma_1}
The pdf of the random variable $W_3$ is given by:
\begin{equation}
    f_{W_3}(t) = \frac{\big(1 - F_{S_{\mathrm{n}}}(t)\big) \lambda_2 e^{-\lambda_2 t}}{1 - M_{S_{\mathrm{n}}}(-\lambda_2)}\, .
    \label{eq_13}
\end{equation}
\end{lemma}
\begin{proof}
    The random variable $W_3$ is the sojourn time in state~$q_1$ prior to a preemption by a negative arrival, and thus has the same distribution as the system time $T_1$ under preemption. Let $H_2 \sim \mathrm{Exp}(\lambda_2)$ denote the inter-arrival time of adversarial packets, and let $\overline{\mathcal{D}}$ denote the event of a negative arrival before normal service completes. Conditioning on $\overline{\mathcal{D}}$, we get:    
    \begin{align*}
    f_{W_3}(t) 
    &= \lim_{\epsilon \to 0} \frac{\Pr\! \big(H_2 \!\in\! [t , t+\epsilon) \mid \overline{\mathcal{D}} \big)}{\epsilon} \\
    &= \lim_{\epsilon \to 0} \frac{\Pr\! \big(H_2 \!\in\! [t , t+\epsilon) \big)\, \Pr\!\big(\overline{\mathcal{D}} \mid H_2 \!\in\! [t, t+\epsilon)\big)}{\Pr(\overline{\mathcal{D}})\, \epsilon} \\
    &= \frac{\Pr(S_{\mathrm{n}\!} > t) \, f_{H_2}(t)}{\Pr(\overline{\mathcal{D}})}  
    = \frac{\big(1 - F_{S_{\mathrm{n}\!}}(t)\big) \lambda_2 e^{-\lambda_2 t}}{1 - M_{S_{\mathrm{n}\!}}(-\lambda_2)} \, ,
    \end{align*}
    where $\Pr(\overline{\mathcal{D}}) \!=\! 1 - M_{S_{\mathrm{n}}}(-\lambda_2)$ follows from evaluating the LST of $S_{\mathrm{n}}$ at $\lambda_2$.
\end{proof}

The inter-departure time between two consecutively delivered updates from source~$s_1$ equals the total sojourn accumulated by the system as it regenerates from state $q_0$ back to $q_0$. This total sojourn time is the sum of the individual sojourns~associated with each state and transitions along all paths of the form $\{q_0, \ldots, q_0\}$. Consequently, $Y_1$ can be~expressed as:
\begin{equation}
    Y_1 = W_1 + W_2 + W_3 + m W_4 + W_5\, ,
    \label{eq_14}
\end{equation}
where $m \geq 0$ counts the number of occurrences of $W_4$ during the cycle from $q_0$ to $q_0$. For example, one possible realization of this cycle is:
\begin{align*}
    \{q_0 \!\to\! q_1,\, q_1 \!\to\! q_2,\, q_2 \!\to\! q_2,\, q_2 \!\to\! q_2,\, q_2 \!\to\! q_1,\, q_1 \!\to\! q_0\}\, ,
\end{align*}
 which occurs with probability $p_{\mathcal{F}}\, p_{\mathcal{D}} (1-p_{\mathcal{D}})(1-p_{\mathcal{F}})^2$. Following the methodology in~\cite{Moltafet2022}, the MGF of $Y_1$ in~\eqref{eq_14}~is thus calculated as:
\begin{align}
    &M_{Y_1}(s) \notag \\
    &\,\,=\Exp\!\left[e^{s Y_1}\right] 
    = \Exp\!\left[e^{s(W_1+W_2+W_3+mW_4+W_5)}\right] \notag \\
    &\,\,=\! \sum_{m=0}^{\infty}\! p_{\mathcal{D}} p_{\mathcal{F}} (1 \!-\!p_{\mathcal{D}\!})(1 \!-\! p_{\mathcal{F}\!})^m 
     \mathbb{E}\!\left[e^{s(W_{1\!}+W_2+W_3+mW_4+W_5)\!}\right]\!, 
    \label{eq_15}
\end{align}
where the MGFs of the individual sojourn times are given by:
\begin{align}
    \mathbb{E}[e^{sW_1}] &= \frac{\lambda_1}{\lambda_1 - s}\, , \label{eq_16} \\
    \mathbb{E}[e^{sW_2}] &= \frac{M_{S_{\mathrm{n}}}(s - \lambda_2)}{M_{S_{\mathrm{n}}}(-\lambda_2)}\, , \label{eq_17} \\
    \mathbb{E}[e^{sW_3}] &= \frac{\lambda_2 \big(1 - M_{S_{\mathrm{n}}}(s - \lambda_2)\big)}{(s - \lambda_2)\big(M_{S_{\mathrm{n}}}(-\lambda_2) - 1\big)}\, , \label{eq_18} \\
    \mathbb{E}[e^{sW_4}] &= \frac{M_{S_{\mathrm{s}}}(s - \Lambda)}{M_{S_{\mathrm{s}}}(-\Lambda)}\, , \label{eq_19} \\
    \mathbb{E}[e^{sW_5}] &= \frac{\Lambda \big(1 - M_{S_{\mathrm{s}}}(s - \Lambda)\big)}{(s - \Lambda) \big(M_{S_{\mathrm{s}}}(-\Lambda) - 1\big)}\, .\label{eq_20} 
\end{align}

Finally, substituting~\eqref{eq_12} and~\eqref{eq_15} into~\eqref{eq_5} and~\eqref{eq_6} yields closed-form expressions for the MGFs of the AoI and PAoI~for source~$s_1$.
\begin{remark}
\label{remark_1}
The $n^{th}$-order moment of the AoI (or PAoI) of source~$s_1$ can be obtained by evaluating the $n^{th}$ derivative~of its MGF at $s=0$. In particular, the first derivative yields the average AoI, $\bar{\Delta}_1$ (or average PAoI,~$\bar{\mathrm{P}}_1$):
\begin{gather}
    \bar{\Delta}_1 = \left.\frac{d M_{A_1}(s)}{ds}\right|_{s=0}, \quad
    \bar{\mathrm{P}}_1 = \left.\frac{d M_{P_1}(s)}{ds}\right|_{s=0}.
    \label{eq_21}
\end{gather}
\end{remark}
\vspace{0.1em}

\subsection{Extension to the Multi-Source Adversarial Model}
\label{sec_3_2}
\fontdimen2\font=0.7ex
The MGF-based framework developed in the preceding sub-section can be directly extended to the case of $N$ sources. Let source~$s_c$ generate adversarial packets at rate~$\lambda_c$, and $\lambda_i$ denote the arrival rate of positive updates from source~$s_i$ $(i \neq c)$.~The aggregate arrival rate of all packets is $\Lambda \!=\! \sum_{k=1}^N \lambda_k$. By applying Lemma~\ref{lemma_1}, the pdf and MGF of the system time~$T_i$~for an update from source~$s_i$ are:
\begin{gather}
    f_{T_{i}}(t) = \frac{f_{S_{\mathrm{n}\!}}(t) \lambda_c e^{-\lambda_c t}}{ M_{S_{\mathrm{n}}}(-\lambda_c)}, \quad 
    M_{T_{i}}(s) = \frac{M_{S_{\mathrm{n}\!}}(s \!-\! \lambda_c)}{M_{S_{\mathrm{n}}}(-\lambda_c)}. \label{eq_22}
\end{gather}
Following the same reasoning as in the two-source case, the MGFs of the AoI and PAoI for source~$s_i$ are given as:
\begin{align}
    M_{A_i}(s) &= \frac{M_{S_{\mathrm{n}}}(s - \lambda_c)\big(M_{Y_i}(s) - 1\big)}{s\,M_{S_{\mathrm{n}}}(-\lambda_c)\,M'_{Y_i}(0)}\, , \label{eq_23}\\
    M_{P_i}(s) &= \frac{M_{S_{\mathrm{n}}}(s - \lambda_c)M_{Y_i}(s)}{M_{S_{\mathrm{n}}}(-\lambda_c)}\, , \label{eq_24}
\end{align}
where $M_{Y_i}(s)$ is the MGF of the inter-departure time for~$s_i$ from state $q_0$ back to $q_0$, obtained by replacing the parameters $\lambda_1$ and $\lambda_2$ in~\eqref{eq_16}-\eqref{eq_20} with $\lambda_i$ and $\lambda_c$, respectively. 
\begin{figure*}[t]
    \centering
  \subfloat[AAoI under Pareto distribution.\label{fig3a}]{%
       \includegraphics[width=0.658\columnwidth]{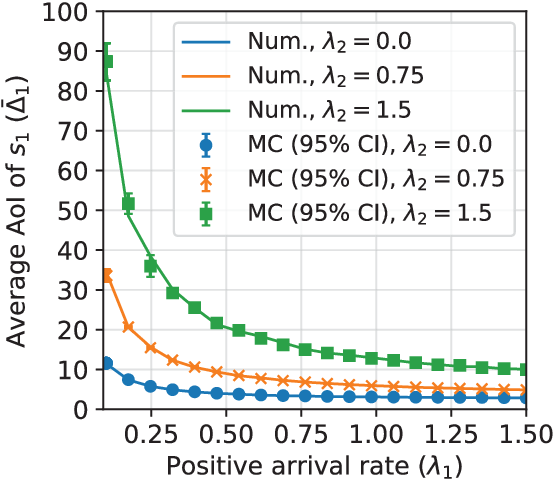}}
    ~
  \subfloat[AAoI under Erlang-2 distribution.\label{fig3b}]{%
        \includegraphics[width=0.658\columnwidth]{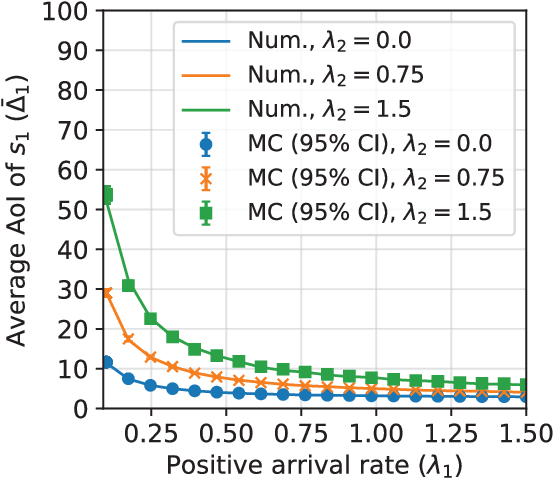}} 
    ~
  \subfloat[AAoI under Hyper-exponential distribution.\label{fig3c}]{%
        \includegraphics[width=0.658\columnwidth]{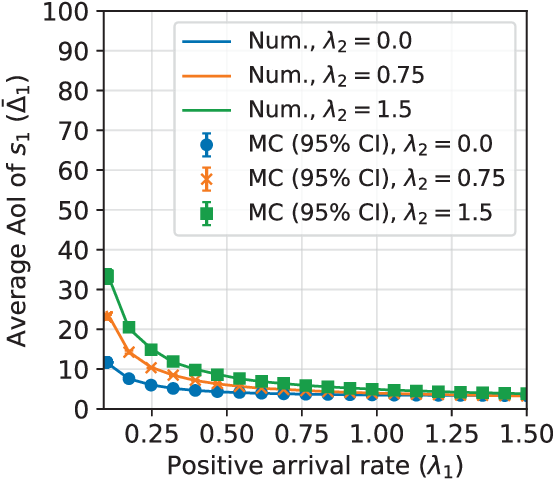}} 
    \vspace{0.2em}
  \caption{AAoI comparison in the two-source M/G/1/1 system with adversarial source~$s_2$ under different service distributions and negative arrival rates~$\lambda_2 \in \{0, 0.75, 1.5\}$: (a) Pareto $(\theta=1,\alpha=3)$; (b) Erlang-2 $(k=2,\mu=4/3)$; (c) Hyper-exponential $(p=0.5,\mu_1=1,\mu_2=0.5)$. Here, $\beta = 1.5$.}
  \label{fig3} 
\end{figure*}
\section{Performance Bounds of AAoI}
\label{sec_4}
\fontdimen2\font=0.7ex
Let $H_i$ denote the inter-arrival time between two consecutive updates from the benign source~$s_i$, $i \!\in\! \{1,2,\dots, N\}$, and let $H_c$ be the inter-arrival time of adversarial arrivals from the compromised source $s_c$ $(i \!\neq\! c)$. From the renewal reward~framework, the time average AoI for source~$s_i$ is given by~\cite{Soysal2021}:
\begin{align}
    \bar{\Delta}_i &= \mathbb{E}[T_{i}] + \frac{\mathbb{E}\big[(Y_i)^2\big]}{2\,\mathbb{E}[Y_i]}\, . \label{eq_25} 
\end{align}

Because the system has no buffer and is subject to preemption by negative arrivals, an update generated by source~$s_i$ may be dropped or delayed, which can only increase the~time between consecutive \emph{delivered} updates. Hence, $Y_i \ge H_i$, which implies that $\mathbb{E}\big[Y_i\big] \ge \mathbb{E}\big[H_i\big]$ and $\mathbb{E}\big[(Y_i)^2\big] \ge \mathbb{E}\big[(H_i)^2\big]$. As a result of this, the following immediate policy-independent \emph{lower bound} on AAoI holds:
\begin{equation}
    \bar{\Delta}_i \;\ge\; \mathbb{E}[T_i] + \frac{\mathbb{E}\big[(H_i)^2\big]}{2\,\mathbb{E}[H_i]}\, .
    \label{eq_26}
\end{equation}

To formulate the worst-case attack model, we consider the adversary~$s_c$ to be limited by: (i) the attack rate limit $(\lambda_c \!\le\! \lambda_{\max})$ and the service slowdown factor $(1 \!<\! \beta \!\leq\! \beta_{\max})$. These constraints reflect practical limitations that prevent the attacker from indefinitely attacking or excessively slowing the server. Let $\mathcal{L}(\lambda_{\max},\beta_{\max})$ be the set of all admissible attack policies satisfying these constraints. Lemma~\ref{lemma_2} postulates the worst-case AAoI \textit{upper bound}.
\begin{lemma}
\label{lemma_2}
For any attack policy $L \!\in\! \mathcal{L}(\lambda_{\max},\beta_{\max})$ and any benign source~$s_i$, where $i \!\in\! \{1,2, \ldots, N\}$ and $i \!\neq\! c$, the age process $A_i(t)$ is stochastically dominated (in the increasing–convex order) by that under the worst-case policy~$\mathcal{P}^\star$ that uses $\lambda_c = \lambda_{\max}$ and $\beta = \beta_{\max}$. Consequently, $\forall  i \neq c$,
\begin{equation}
    \bar{\Delta}_i \leq \bar{\Delta}_i^{\mathcal{P}^\star\!} = \mathbb{E} [T_i] + \frac{\mathbb{E} \big[(Y_i^{\mathcal{P}^{\star}})^2\big]}
           {2\,\mathbb{E} \big[Y_i^{\mathcal{P}^{\star}}\big]}\, ,
    \label{eq_27}
\end{equation}
where $Y_i^{\mathcal{P}^{\star}\!}$ includes the slow-service effect and its moments are obtained from the MGFs derived in Section~\ref{sec_3_2} by~substituting $\lambda_c \!=\! \lambda_{\max}$ and $M_{S_{\mathrm{s}}}(s) \!=\! M_{S_{\mathrm{n}}}(\beta_{\max}s)$.
\end{lemma}
\begin{proof}
Increasing $\lambda_c$ raises negative preemption and the frequency of slow-service state entries, while increasing~$\beta$ lengthens slow-state sojourns. Both operations make $Y_i$ larger in the increasing–convex order and increase $\mathbb{E}[T_i]$. Since~$\bar{\Delta}_i$ in \eqref{eq_25} is monotone in $\mathbb{E}[T_i]$, $\mathbb{E}[Y_i]$, and $\mathbb{E}\big[(Y_i)^2\big]$, the supremum over $\mathcal{L}(\lambda_{\max},\beta_{\max})$ is attained at $\lambda_c \!=\! \lambda_{\max}$ and $\beta \!=\! \beta_{\max}$, yielding the upper bound in \eqref{eq_27}. 
\end{proof}

Note that the lower bound in~\eqref{eq_26} depends only on the benign source’s arrival process via $H_i$, thus ignoring adversarial preemptions and service slowdowns, whereas the upper bound in~\eqref{eq_27} captures the worst admissible impact of~adversarial preemption and slowdowns through $(\lambda_{\max},\beta_{\max})$.

\begin{figure*}[t]
    \centering
  \subfloat[AAoI under Pareto distribution.\label{fig4a}]{%
       \includegraphics[width=0.668\columnwidth]{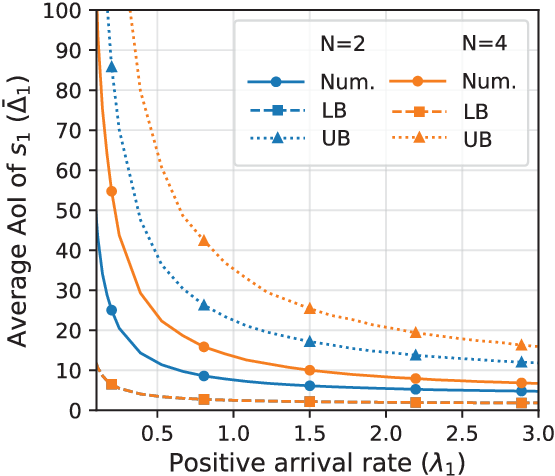}}
    ~
  \subfloat[AAoI under Erlang-2 distribution.\label{fig4b}]{%
        \includegraphics[width=0.668\columnwidth]{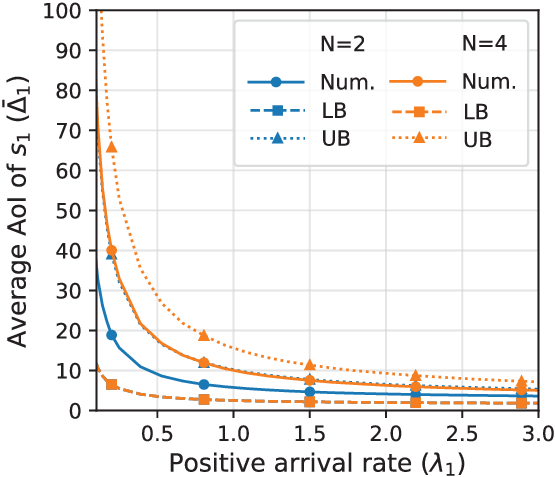}} 
    ~
  \subfloat[AAoI under Hyper-exponential distribution.\label{fig4c}]{%
        \includegraphics[width=0.668\columnwidth]{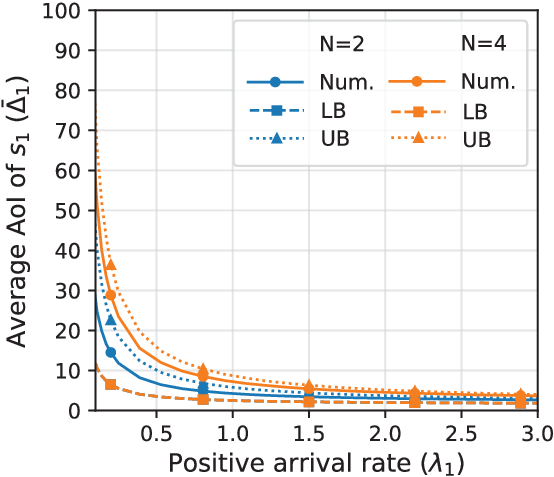}} 
    \vspace{0.2em}
  \caption{AAoI bound comparisons in the M/G/1/1 system with $N \in \{2,4\}$ sources and under different service distributions: (a) Pareto $(\theta=1,\alpha=3)$; (b) Erlang-2 $(k=2,\mu=4/3)$; (c) Hyper-exponential $(p=0.5,\mu_1=1,\mu_2=0.5)$. For Num., $\lambda_c=1$ and $\beta = 1.5$.}
  \label{fig4} 
\end{figure*}
\section{Numerical Results and Discussions}
\label{sec_5}
In this section, we evaluate the achievable AAoI of the proposed adversarial model under Pareto (Par.), Erlang-2~(Erl.), and Hyper-exponential order-2 (Hyp.) service time distributions~\cite{Moltafet2022, Kumar2025}. Both numerical (Num.) evaluations and Monte Carlo~(MC) simulations (averaged over 12 runs) are conducted, and the results are~benchmarked against the baseline no-adversary case, i.e., $\lambda_c \!=\! 0$. For the Par. distribution,~we use scale parameter~$\theta \!=\! 1$ and shape $\alpha \!=\! 3$. For the Erl. distribution, we set shape $k \!=\! 2$ with rate $\mu \!=\! 4/3$, and for the Hyp. distribution, we adopt rates ${(\mu_1,\mu_2) \!=\! (1, 0.5)}$ (mean normalized to $1/\mu$) and balanced mixture probability $p \!=\! 0.5$. These parameters are chosen so that the mean service time is the same across all three distributions, i.e., $\Exp [S]=1.5$. Unless otherwise stated, we set $\lambda_{\max} \!=\! \beta_{\max} \!=\! 2$.

\figurename{~\ref{fig3}} depicts the AAoI of the benign source~$s_1$ under three service time distributions in the two-source setting. The results show that the AAoI is highest under the Par. distribution with $\Exp[S] \!=\! \alpha\theta/(\alpha-1)$, moderate under the Erl. distribution with $\Exp[S] \!=\! k/\mu$, and lowest under the Hyp. distribution with $\Exp[S] \!=\! p/\mu_1+(1-p)/\mu_2$. This behavior stems from~the distinct tail characteristics of the service times. The Par. distribution shown in \figurename{~\ref{fig3a}} exhibits a heavy tail with a decreasing hazard rate. This implies that once a long service is encountered, it is likely to remain long. Such prolonged service times are highly susceptible to adversarial preemptions, which extend~the inter-delivery intervals of the positive updates and result in higher AoI values. For Erl. service (\figurename{~\ref{fig3b}}), the distribution is more concentrated around its mean with limited variability, so updates are less likely to complete unusually early and thus, leading to higher AoI relative to the Hyp. distribution~shown in \figurename{~\ref{fig3c}}. As for the Hyp. case, a non-negligible probability mass is concentrated at very short service durations, which substantially increases the likelihood that an update completes service before a negative arrival interrupts it, thereby reducing AoI. Hence, even under the same $\mathbb{E}[S]$, the service distribution significantly influences system timeliness, with lighter-tailed and more variable distributions favoring lower AAoI in adversarial system.

\figurename{~\ref{fig4}} compares the AAoI of $s_1$ versus $\lambda_1$ under the bounds derived in~\eqref{eq_26} and \eqref{eq_27}. The \emph{lower bound} (LB), which is computed with no adversary and no slowdown ($\lambda_{c} \!=\! 0$, $\beta \!=\! 1$), decays rapidly in $\lambda_1$ because the inter-arrival shrinks~as updates are generated more frequently. Hence, \eqref{eq_26} reduces to $\mathbb{E}[S]+1/\lambda_1$. The \emph{upper bound} (UB), evaluated at the worst case ($\lambda_{c} \!=\! \lambda_{\max}$, $\beta \!=\! \beta_{\max}$), is the highest curve which reflects more frequent negative preemptions and slower service thus, inflating the renewal second-moment term in~\eqref{eq_27}. The Num. curve ($\lambda_{\text{c}} \!=\! 1$, $\beta \!=\! 1.5$) lies between LB and UB across all~$\lambda_1$ values, validating the bounding construction. Increasing the number of sources from $N \!=\! 2$ to $N \!=\! 4$ shifts all curves~upward. Evidently, extra contention raises the busy probability and the chance that status updates from $s_1$ are dropped or exposed to adversarial preemption, so both $T_1$ and $Y_1$ tend~to grow. The heavy-tailed Par. distribution in~\figurename{~\ref{fig4a}} yields the largest AAoI and widest LB-UB gap because occasional very long services are highly vulnerable to preemption/slow epochs. The more concentrated Erl. distribution in~\figurename{~\ref{fig4b}} moderates this effect, whereas the Hyp. distribution in~\figurename{~\ref{fig4c}} provides the smallest AAoI since a non-trivial number of very short services often completes before a negative arrival, lowering~$Y_1$ and in turn, the average age.

\section{Conclusion}
\label{sec_6}
\fontdimen2\font=0.66ex
In this paper, we investigated the impact of adversarial attacks on the AAoI in multi-source M/G/1/1 status updating systems. Closed-form MGFs for AoI and PAoI were~derived to establish performance bounds under adversarial constraints on maximum attack rate and service slowdown. The analysis showed that the worst-case adversarial strategy yields a tractable characterization of the system’s performance limits, while numerical and Monte Carlo evaluations confirmed the theoretical results and highlighted the tradeoff between adversarial intensity and service efficiency across different service distributions. Future directions include extending the model to buffer-aided queues and exploring joint optimization of energy expenditure and AoI in adversarial IoT networks.

\ifCLASSOPTIONcaptionsoff
  \newpage
\fi



%
\bibliographystyle{IEEEtran}
\bibliography{IEEEabrv, myref}

%

\begin{IEEEbiography}{Michael Shell}
Biography text here.
\end{IEEEbiography}

\begin{IEEEbiographynophoto}{John Doe}
Biography text here.
\end{IEEEbiographynophoto}


\begin{IEEEbiographynophoto}{Jane Doe}
Biography text here.
\end{IEEEbiographynophoto}




\end{document}